\documentclass[conference,9pt]{IEEEtran}
\usepackage{algorithmic}
\def\BibTeX{{\rm B\kern-.05em{\sc i\kern-.025em b}\kern-.08em
    T\kern-.1667em\lower.7ex\hbox{E}\kern-.125emX}}
\usepackage[pdfencoding=auto, psdextra]{hyperref}
\usepackage{amsmath,amsfonts,amsthm}
\usepackage{mathtools}
\usepackage{tikz}
\usetikzlibrary{angles, quotes, patterns, backgrounds}  % Libraries in standalone file
\usepackage{float}
\usepackage{placeins}
\usepackage{adjustbox}

\makeatletter
\def\fixedlabel#1#2{%
  \@bsphack%
  \protected@write\@auxout{}%
         {\string\newlabel{#1}{{#2}{\thepage}}}%
  \@esphack}
\makeatother

\newcommand{\bs}[1]{{#1}}

\usepackage{graphicx} % Required for inserting images

\newcommand{\matr}[1]{\bm{#1}}
\usepackage{bm}
\usepackage{cite}
\newcommand{\vect}[1]{\bm{#1}}
\theoremstyle{definition}
\newtheorem{theorem}{Theorem}
\newtheorem{corollary}{Corollary}[theorem] % Use theorem counter as `parent`

\newtheorem{lemma}[]{Lemma}

\usepackage{dsfont}
\usepackage{algorithmic}
\usepackage{calc}
\usepackage{array}
\usepackage{tabularx}
\usepackage{physics}

\usepackage{makecell}

\usepackage[font=normalsize,labelfont=sf,textfont=sf, justification=centering]{subcaption}
\usepackage{textcomp}
\usepackage{stfloats}
\usepackage{url}
\usepackage{verbatim}
\usepackage{graphicx}
\usepackage{xcolor}
\usepackage{balance}
\usepackage[pdfencoding=auto, psdextra]{hyperref}
\usepackage{booktabs}
\usepackage{cases}

\usepackage{xstring}
\DeclareMathSymbol{\Theta}{\mathord}{operators}{"02}
\newcommand{\set}[1]{\mathcal{#1}}

\newcommand{\sqnormvec}[1]{\left|\left|#1\right|\right|^{2}_{2}}
\NewDocumentCommand{\msub}{o o m}
{
    \IfValueT{#2}{\left[#3\right]_{\set{#1}, \set{#2}}}
    \IfNoValueTF{#1}{#3}{\left[#3\right]_{\set{#1}}}
}
\NewDocumentCommand{\msubgen}{o o m}
{
    \IfValueT{#2}{\left[#3\right]_{#1, #2}}
    \IfNoValueTF{#1}{#3}{\left[#3\right]_{#1}}
}
\NewDocumentCommand{\matrsub}{ o o m }
{
    \IfValueT{#2}{[\matr{#3}]_{\set{#1}, \set{#2}}}
    \IfNoValueTF{#1}{\matr{#3}}{[\matr{#3}]_{\set{#1}}}
}
\NewDocumentCommand{\matrsubpow}{ o o m m}
{
    \IfValueT{#2}{[\matr{#3}^{#4}]_{\set{#1}, \set{#2}}}
    \IfNoValueTF{#1}{\matr{#3}^{#4}}{[\matr{#3}^{#4}]_{\set{#1}}}
}
\NewDocumentCommand{\vectsub}{ o o m }
{
    \IfValueT{#2}{[\vect{#3}]_{\set{#1}, \set{#2}}}
    \IfNoValueTF{#1}{\vect{#3}}{[\vect{#3}]_{\set{#1}}}
}
\NewDocumentCommand{\projbl}{ o o }
{
    \IfValueT{#2}{[\matr{\Pi}_{bl(\set{K})}]_{\set{#1},\set{#2}}}
    \IfNoValueTF{#1}{\matr{\Pi}_{bl(\set{K})}}{[\matr{\Pi}_{bl(\set{K})}]_{\set{#1}}}
}
\NewDocumentCommand{\proj}{ o o m }
{
    \IfValueT{#2}{[\matr{\Pi}_{\set{#3}}]_{\set{#1}, \set{#2}}}
    \IfNoValueTF{#1}{\matr{\Pi}_{\set{#3}}}{[\matr{\Pi}_{\set{#3}}]_{\set{#1}}}
}
\newcommand{\matrsubU}[1]{\matrsub[#1, K]{U}}

\NewDocumentCommand{\expect}{ o m }
{
    \IfNoValueTF{#1}{{\mathbb{E}}\left[{{#2}}\right] }{{\mathbb{E}}_{#1}\left[{{#2}}\right]}
}
\NewDocumentCommand{\prob}{ o m }
{
    \IfNoValueTF{#1}{{\mathbb{P}}\left({{#2}}\right) }{{\mathbb{P}}_{#1}\left({{#2}}\right)}
}
\NewDocumentCommand{\expectnobig}{ o m }
{
    \IfNoValueTF{#1}{{\mathbb{E}}[{{#2}}] }{{\mathbb{E}}_{#1}[{{#2}}]}
}

\begin{document}

\title{On the Impact of Downstream Tasks on Sampling and Reconstructing Noisy Graph Signals}

\author{\IEEEauthorblockN{Baskaran Sripathmanathan,} \and \IEEEauthorblockN{Xiaowen Dong,} \IEEEauthorblockA{University of Oxford} \and \IEEEauthorblockN{Michael Bronstein} }
%\author{\IEEEauthorblockN{Anonymous Authors} \IEEEauthorblockA{Anonymous Institution}}
\maketitle

%%%Copyright Notice
\makeatletter
\def\ps@IEEEtitlepagestyle{
  \def\@oddfoot{\mycopyrightnotice}
  \def\@evenfoot{}
}
\def\mycopyrightnotice{
  {\footnotesize
  \begin{minipage}{\textwidth}
  \centering
  \copyright 2025 IEEE. Personal use of this material is permitted. Permission from IEEE must be obtained for all other uses, in any current or future media, including reprinting/republishing this material for advertising or promotional purposes, creating new collective works, for resale or redistribution to servers or lists, or reuse of any copyrighted component of this work in other works.
  \end{minipage}
  }
}

\begin{abstract}
We investigate graph signal reconstruction and sample selection for classification tasks. We present general theoretical characterisations of classification error applicable to multiple commonly used reconstruction methods, and compare that to the classical reconstruction error. We demonstrate the applicability of our results by using them to derive new optimal sampling methods for linearized graph convolutional networks, and show improvement over other graph signal processing based methods.
\end{abstract}


\section{Introduction}
\iffalse
\begin{itemize}
    \item signal processing and machine learning on graphs are increasing common with many applications
    \item data are often noisy and sometimes even incomplete; while the SP community has addressed this via extending classical sampling and reconstruction to graph setting, the ML community hasn't produced much theoretical results in the setting of learning with missing data (as well as its relation to the SP case, due to the difference in ``downstream'' tasks) 
    \item to fill this gap, we provide theoretical results on classification error in presence of missing data, and compare that with results for sampling \& reconstruction (emphasising the trade-off between the two)
    \item Based on the theoretical results, we present a sampling scheme for classification, and show its effectiveness in both the classical setting and by linearly approximating a GCN
    \item [optional] compare empirically how the optimal sample set in one compares to that in the other (based on characteristics of the sample set)
    \item results on real-world data
\end{itemize}

Questions:
\begin{itemize}
    \item Do I need to compare it to GCN methods for missing data?
\end{itemize}
\fi
Signal processing and machine learning on graphs have become increasingly common, with many applications such as analysing and predicting brain fMRI signals \cite{itani2021graph}, urban air pollution \cite{jain2014big}, political preferences \cite{renoust2017estimating} and \bs{protein function \cite{szklarczyk2019string}}. However, data defined on a network domain, often called `graph signals', is frequently noisy and sometimes even incomplete, %-> this is a general paragraph about context and problem statement 
making such tasks challenging.

\iffalse
%(to be moved to paragraphs below) 
The signal processing community has addressed this by generalising classical sampling and reconstruction tools to the graph setting, with a focus on efficient reconstruction of signals via sample choice or reconstruction method \cite{ortega2018graph}. Theoretical results have also been derived on the performance of signal reconstruction and sample choice under reconstruction loss \cite{pesenson2008sampling, puy2018random, chamon2016near, chamon2017greedy, shomorony2014sampling}. However, there is a significant lack of investigation of the impact of incomplete data on downstream tasks such as classification. The machine learning literature has so far only been centered around developing empirical frameworks in this setting \cite{rossi2021unreasonable}, while theoretical results on classification under this missing data on graphs setting are scarce. 
%While the machine learning literature on this setting is centered around empirical results \cite{rossi2021unreasonable}, the signal processing community has generalised classical sampling and reconstruction tools to the graph setting, with a focus on efficient reconstruction of signals via sample choice or reconstruction method \cite{wang2018optimal}[cite]. However, theoretical results in graph signal processing have mainly focused on the performance of signal reconstruction and sample choice under reconstruction loss [cites], rather than with the loss induced by more complex downstream tasks such as classification with a GCN.
\fi

The graph signal processing (GSP) community has addressed the issue of noisy and incomplete data by generalising the tools of sampling and reconstruction from classical signal processing, via extending the classical shift operator to a graph shift operator \cite{ortega2018graph} such as the Laplacian. The majority of studies on graph-based sampling focus on designing efficient sampling schemes for criteria which approximate \cite{wang2018optimal, wang2019low, mfn, tremblay2017graph, jayawant2021doptimal,puy2018random, pakiyarajah2025graph} or bound \cite{bai2020fast,eoptimalchen} reconstruction loss, grounded in the optimal design of experiments \cite{pukelsheim2006optimal}, \bs{with some work simultaneously learning the underlying graph topology \cite{zhang2025graph, berger2020efficient, sridhara2024towards}}.  Theoretical results have also been derived on the performance of sample choice and signal reconstruction under reconstruction loss \cite{pesenson2008sampling, puy2018random, chamon2016near, chamon2017greedy, shomorony2014sampling, chen2016signal}. %with those on perfect reconstruction \cite{pesenson2008sampling} generalising to other downstream tasks, as perfect reconstruction of the data implies perfect reconstruction of any function of the data. 
%However, those results require a bandlimited signal model and computationally expensive reconstruction, and are sensitive to these assumptions \cite{sripathmanathan2024impact}, limiting their applicability. 
\bs{However, these results do not immediately carry over to other tasks such as classification, as they rely on the specific structure of reconstruction loss -- e.g., the proof of near-optimality of greedy sampling on graphs \cite{chamon2016near} relies on the loss's (approximate) supermodularity. Thus, there is a significant lack of investigation in the GSP literature of the impact of incomplete data on downstream learning tasks.} %with Graph Neural Networks (GNNs).}

\iffalse
However, as these works focus on reconstruction error, the theoretical results only apply specifically to this setting -- e.g., the results on the near-optimality of greedy sampling \cite{chamon2016near} depend on the (approximate) supermodularity of the specific reconstruction loss, and hence do not immediately apply to the classification setting \bs{which has} a different goal and error characterization. Except for results on perfect reconstruction in the noiseless setting \cite{pesenson2008sampling}, theoretical results on reconstruction in the broader settings \cite{puy2018random,chamon2016near, chamon2017greedy,shomorony2014sampling, sripathmanathan2024impact} do not immediately carry over to the classification setting. Thus, there is a significant lack of investigation in the GSP literature of the impact of incomplete data on downstream tasks such as classification.
\fi
\bs{
The graph machine learning (GML) community, on the other hand, does address the issue of incomplete data, often in the context of downstream tasks such as node classification. They take two main approaches -- label prediction and feature imputation. Label prediction sidesteps the issue of missing data by directly predicting the missing node class labels, such as by label propagation \cite{zhou2003learning} or by using graph neural networks (GNNs) on the partial data with modified message passing \cite{jiang2020incomplete, chen2020learning}. Feature imputation involves reconstructing the graph signals, potentially alongside missing edges, usually as a pre-processing step to inference. Such methods include matrix completion methods \cite{kalofolias2014matrix} and autoencoder-based methods \cite{spinelli2020missing}. Another example is feature propagation \cite{rossi2021unreasonable}, which is a graph Laplacian-based feature imputation method that applies a GSP-motivated approach to downstream tasks such as classification with GNNs. Finally, some approaches combine label prediction and feature imputation in an end-to-end training framework \cite{taguchi2021graph, you2020handling}. While work in GML engages with missing data in settings with varied downstream tasks, it has done so through an empirical lens -- evaluating the performance of the proposed methods on specific datasets -- and thus theoretical results remain scarce.}

\iffalse
The graph machine learning (GML) community, on the other hand, takes two approaches to addressing the problem of missing data -- feature imputation and label prediction. Feature imputation involves reconstructing the graph signals, \bs{potentially along with missing edges}, and such methods include matrix completion methods \cite{kalofolias2014matrix} \iffalse \cite{candes2012exact,cai2010singular}\fi and autoencoder-based methods \cite{spinelli2020missing}. Label prediction \bs{sidesteps the issue of missing data by} directly predicting the missing node class labels, such as \bs{by} label propagation \cite{zhou2003learning} or directly using GNNs on the partial data with adjusted message passing \cite{jiang2020incomplete, taguchi2021graph, chen2020learning}. Some approaches combine the two \cite{you2020handling}. We note that Feature Propagation \cite{rossi2021unreasonable} makes inroads into bridging the GSP and GML approaches by feature imputation motivated by a graph Laplacian-based criterion while evaluating its effectiveness under a more complex downstream task. While work in GML engages with missing data in settings with significantly more complex downstream tasks than graph signal processing, it has done so through an empirical lens -- evaluating the performance of the proposed methods on specific datasets -- and thus theoretical results remain scarce. %Our work attempts to fill this gap.
%a paragraph on summary of II.A: work in SP and limitations 

%a paragraph summary of II.B: work in ML and limitations 
\fi

The separation of the two lines of work above points to an important aspect of dealing with incomplete data: sampling and reconstruction should be designed to serve the downstream task at hand and, depending on the task, the ``reconstruction error'' we care about might be different. For example, \bs{with recommender systems with incomplete participant data, our goal is to make effective recommendations (which can be formulated as a classification task) rather than reconstructing the attributes of the participants.}

In this paper, we fill this gap in the literature by providing a theoretical characterisation of the performance of binary classification \bs{with a linearized graph convolutional network (GCN)} under missing data and noisy observation (Theorem \ref{thm:general_class_err}). We provide a bound showing the relationship between classification and reconstruction error (Corollary \ref{corr:relationship}), and compare the change in classification versus reconstruction error as a function of the sample size both theoretically and empirically. %We present a diagram showing the relationship under our model (Fig. \ref{fig:class_rec_rel_fig}), and compare the change in classification error versus reconstruction error as a function of the sample size both theoretically and empirically. 
%comparison between classification and reconstruction
\bs{We show that optimal sampling derived in the reconstruction setting can underperform random sampling when applied to classification tasks}. We then apply our theoretical results to derive a novel sample selection scheme for classification under a linearized GCN, and demonstrate that it improves upon using sample selection schemes which are optimal in the reconstruction setting. %We provide empirical experiments demonstrating how mean classification error and reconstruction error vary under sample size, under both synthetic and real world data.

\section{Background} %\& Problem Formulation}
%\subsection{Problem Formulation}
%In this paper, we aim to show a quantitative relationship between the classification and reconstruction error for graph signals with missing data, and the implications of their differences. We begin by providing the necessary technical background and then present our problem setting.

\subsection{Graphs and graph signals}
A graph $\mathcal{G}$ consists of a set $\set{V}$ of $N$ vertices, a set of edges $\set{E}$ between these vertices and a corresponding set of edge weights. We assume that $\mathcal{G}$ is connected and undirected. A graph signal $\vect{x}$ is a real-valued function $\vect{x}:\set{V} \to \mathbb{R}$, equivalently written as a vector in $\mathbb{R}^{N}$ by having each component correspond to a vertex.
We define \bs{the adjacency matrix $\matr{A}_{ij} = \mathds{1}\{(i,j) \in \set{E}\}$, the degree matrix $\matr{D}_{ij} = \text{deg}(i)\mathds{1}\{i = j\}$, the normalized adjacency matrix $\tilde{\matr{A}} = \matr{D}^{-\frac{1}{2}}\matr{A} \matr{D}^{-\frac{1}{2}}$ and the normalized augmented adjacency matrix $\tilde{\matr{A}}_{\gamma} = (\matr{D} + \gamma\matr{I})^{-\frac{1}{2}}(\matr{A} + \gamma\matr{I})(\matr{D} + \gamma\matr{I})^{-\frac{1}{2}}$ where $\matr{I}$ is the identity matrix. This augmentation can be understood as adding self-loops to $\mathcal{G}$, and note adding no loops yields $\tilde{\matr{A}}_{0} = \tilde{\matr{A}}$.}
%We define the degree-normalized adjacency matrix as $\tilde{\matr{A}}_{ij} = \frac{1}{\sqrt{\text{deg}(i)\text{deg}(j)}}\mathds{1}\left\{(i,j) \in \set{E}\right\}$. \bs{We define the normalized augmented adjacency matrix $\tilde{\matr{A}}_{loopy}$ to be $\tilde{\matr{A}}$ of $\mathcal{G}$ with added self-loops.}
Throughout this paper, we consider the symmetrically normalized Laplacian $\matr{L} = \matr{I} - \tilde{\matr{A}}$.

\subsection{Sampling of graph signals}
\bs{For any matrix $\matr{X}$,} we define $\matrsub[A,B]{X}$ to be the submatrix of $\matr{X}$ with row indices in $\set{A}$ and column indices in $\set{B}$. Similarly, we write $\vectsub[A]{x}$ for the subvector of $\vect{x}$ with indices in $\set{A}$. We let $\set{N} = \{1,\ldots,N\}, \bs{\set{D} = \{1,\ldots, d\}}$ and $\set{K} = \{1,\ldots,k\}$, and \bs{for any index set $\set{S}$, we write $\set{S}^{C} = \{ i \in \set{N} \mid i \notin \set{S}\}$}.
We write $\matr{L} = \matr{U}\matr{\Lambda}\matr{U}^{T}$ to be the eigendecomposition of $\matr{L}$, so the columns of $\matr{U}$ correspond to the eigenvectors of $\matr{L}$, ordered by eigenvalue in increasing order (all of which are nonnegative). We define $\projbl = \matrsubU{N}\matrsubU{N}^{T}$ to be the projection matrix which makes a signal $k$-bandlimited, i.e., it acts as a ideal bandpass filter, keeping only the components corresponding to the $k$ lowest eigenvalues of $\matr{L}$.

\subsection{Reconstruction Methods}
A \emph{reconstruction method} takes (potentially) noisy observations of a graph signal on a node sample set $\set{S}$.
We consider linear methods such as Least Squares (LS) or Feature Propagation (FP) \cite{rossi2021unreasonable}. This means that for a fixed sample set $\set{S}$, we can write the reconstruction operator as a matrix $\matr{R}_{\set{S}}:\mathbb{R}^{|\set{S}|} \to \mathbb{R}^{N}$: %\bs{, where for LS $\matr{R}_{\set{S}} = \matrsubU{N}\matrsubU{S}^{\dagger}$ and for FP $\matr{R}_{\set{S}} = \matrsub[N,S]{I} + \msubgen[\set{N},\set{S}^{C}]{\matr{I}}\msubgen[\set{S}^{C}]{\matr{L}}^{-1}\msubgen[\set{S}^{C},\set{S}]{\matr{L}}$.}
%. They have the following forms:
    \begin{align}
    \textrm{LS:\quad} \matr{R}_{\set{S}} &= \matrsubU{N}\matrsubU{S}^{\dagger} \label{eq:defn_RS:LS}\\
    %\textrm{GLR:\quad} \matr{R}_{\set{S}} &= \msub[N,S]{( \proj{S} + \mu \matr{L})^{-1}} \label{eq:defn_RS:GLR} \\
    \textrm{FP:\quad} \matr{R}_{\set{S}} &= \matrsub[N,S]{I} + \msubgen[\set{N},\set{S}^{C}]{\matr{I}}\msubgen[\set{S}^{C}]{\matr{L}}^{-1}\msubgen[\set{S}^{C},\set{S}]{\matr{L}}.
    \end{align}

\subsection{Problem Setting}
\label{sec:problem_setting}
%We assume the following setup throughout: 
We assume the following setup for \bs{all} results: %our theoretical and experimental results: 
\begin{itemize}
    \item A feature matrix $\matr{X}\in\mathbb{R}^{N \times d}$, where the columns are \bs{jointly Gaussian} graph signals that each follow $\mathcal{N}\left(\vect{0},\matr{\Sigma}\right)$. We assume a smooth signal model (e.g. $\matr{\Sigma} = \projbl$ \cite{sripathmanathan2024impact} or $\matr{L}^{\dagger}$ \cite{dong2016learning});
    \item Binary labels $\matr{l} \in \{-1,1\}^{N}$ generated by a known function $f:\, \mathbb{R}^{N \times d}\to \mathbb{R}^{N}$  (i.e., $\vect{l} = \text{sign}\left(f\left(\matr{X}\right)\right)$ with sign elementwise);
    \item Observations of $\matr{X}$ at a sample set $\set{S}$ corrupted by white noise of variance $\eta^{2}$ (i.e., $\msubgen[\set{S},\set{D}]{\matr{X} + \eta \cdot \matr{\epsilon}}$); %where $\matr{\epsilon}_{ij} \sim N(0,1)$ are i.i.d.
    %epsilon is the variance of
    %White noise $\vect{\epsilon}\in \mathbb{R}^{N \times d}$ where each entry is i.i.d $\mathcal{N}(0,1)$;
    \item A linear reconstruction $\hat{\matr{X}}$ of the feature matrix, such as LS or FP, giving $\hat{\matr{X}} = \matr{R}_{\set{S}} \msubgen[\set{S},\set{D}]{\matr{X} + \eta \cdot \matr{\epsilon}}$.
    %\item Labels $\hat{\vect{l}} \in \{-1,1\}^{N}$ inferred from the reconstructed features (i.e., $\hat{\vect{l}} = \text{sign}(f(\hat{\matr{X}}))$ ).
    %\item `Reconstructed' labels $\hat{\vect{l}} = \text{sign}(f(\hat{\matr{X}}))$.
\end{itemize}

We consider the following losses:
\begin{align}
    %\text{Classification Loss} &= \expect{\# \text{ misclassifications}}, \\
    &\bs{\text{Classification Loss}} = {\sum_{i \in \set{V}}\prob{\text{sign}\left(f(\matr{X})_{i}\right) \neq \text{sign}(f(\hat{\matr{X}})_{i})}} \\
    &\text{Reconstruction Loss} = \expect{||{\vect{X} - \hat{\vect{X}}}||^{2}_{F}}.
    %\text{Classification Loss} &= \expect{\# \text{ misclassifications}}.
\end{align}

%\bs{TODO here $f(X_i)$ means}

\bs{The classification loss defined above is the sum of per-node misclassification probabilities. We refer to $f(\matr{X}) \in \mathbb{R}^{N}$ and $f(\hat{\matr{X}}) \in \mathbb{R}^{N}$ as vectorial `clean output' and `reconstructed output', respectively, and $f(\matr{X}) - f(\hat{\matr{X}})$ as the `output error'. We let $f(\matr{X})_{i}$ denote the component of $f(\matr{X})$ at node $i$. %Our results also apply to the trivial GCN $f(\matr{X})=\matr{X}$, which gives a simple extension of the GSP reconstruction setting.
}

%The smoothness of the features with regards to the graph comes from choice of $\matr{\Sigma}$, such as $\projbl$ \cite{sripathmanathan2024impact} or $\matr{L}^{\dagger}$ \cite{dong2016learning}.

%Notably, we assume that we know the `true' function $f$ generating labels from features. While in practice we might approximate this function by training a GCN, we wish to focus on the impact of sample size and reconstruction method on classification performance, so we assume $f$ is known.

\section{Theoretical Results and Implications}
\subsection{General Results}
\label{sec:Gen_Theory}
%In this section, we present results for features drawn from a zero-mean distribution. 
We begin with a core lemma:

\begin{lemma}
\label{lemma:gaussian_sign}
Consider \bs{$X,Y$ real scalar random variables with} $\binom{X}{Y} \sim \mathcal{N}\left( \vect{0}, \begin{psmallmatrix} \sigma^{2} & c \\ c & \nu^{2} \end{psmallmatrix} \right)$ with $\sigma, \nu, c \in \mathbb{R}$. Let $\rho = \text{Corr}(X,Y) = \frac{c}{\sigma\nu}$. Then
\begin{equation}
    %\mathbb{P}\left( X \text{ and } Y \text{ have different signs}\right) = \frac{ \text{arccos }\rho }{\pi}.
     \bs{\prob{\text{sign}(X) \neq \text{sign}(Y) }} = \frac{ \text{arccos }\rho }{\pi}.
\end{equation}
    
\end{lemma}
\begin{proof}
 Let $X' = \frac{X}{\sigma}, Y' = \frac{Y}{\nu}$, then $\binom{X'}{Y'} \sim \mathcal{N}\left( \vect{0}, \begin{psmallmatrix} 1 & \rho \\ \rho & 1 \end{psmallmatrix} \right)$. Note that $\text{sign}(X) \neq \text{sign}(Y) \iff XY < 0 \iff X'Y' < 0$; therefore $\prob{\text{sign}(X) \neq \text{sign}(Y) } = \prob{X'Y' < 0}$. \bs{Therefore we only need to prove the lemma for $X'$ and $Y'$, i.e., when $\sigma=\nu=1$ and $c=\rho$.}

\iffalse
Let $r = \frac{\rho}{\sqrt{1-\rho^2}}$ and $Z' = \frac{1}{\sqrt{1-\rho^{2}}}(\rho X' - Y')$. Then $\binom{X'}{Z'} \sim \mathcal{N}\left(\vect{0}, \matr{I}\right)$ and $\prob{X'Y'<0} = \prob{rX'^{2} < X'Z'} = \prob{\{rX' < Z' \text{ and } X' > 0\} \cup \{rX' > Z' \text{ and } X' < 0 \}}$. We depict this set as the shaded area in Figure \ref{fig:proof_fig}. 

\begin{figure}[H]
    \centering
    \input{tikz/easy-proof.tikz}
    \caption{$\{rX'^{2} < X'Z' \}$}
    \label{fig:proof_fig}
\end{figure}
\fi
\begin{figure}[H]
    \begin{minipage}[T]{0.48\columnwidth}
    %\vspace{0pt}
        \raggedright
        Let $r = \frac{\rho}{\sqrt{1-\rho^2}}$ and $Z' = rX' - \sqrt{1+r^{2}} Y'$ \iffalse $Z' = \frac{1}{\sqrt{1-\rho^{2}}}(\rho X' - Y')$\fi. 
        Then $\binom{X'}{Z'} \sim \mathcal{N}\left(\vect{0}, \matr{I}\right)$ and 
        $\prob{X'Y'<0}$
        \noindent$ = \prob{rX'^{2} < X'Z'} $ 
        \noindent$= 
        \mathbb{P}(\{rX' < Z' \text{ and } X' > 0\} \cup \{rX' > Z' \text{ and } X' < 0 \})$. 
        We depict this set as the shaded area in Figure \ref{fig:proof_fig}.
    \end{minipage}
    \hfill
    \begin{minipage}[T]{0.48\columnwidth}
        \centering
        \begin{tikzpicture}[scale=0.35, background layer/.style={fill=none}]
% Define constants
\def\rhovalue{0.8}
\pgfmathsetmacro{\slope}{\rhovalue/sqrt(1-\rhovalue^2)}
\def\endX{5}
\pgfmathsetmacro{\endZ}{\slope*\endX}

% Draw axes
\draw[->] (0,-3) -- (0,3) node[left] {$Z'$};
\draw[->] (-4,0) -- (4,0) node[below right] {$X'$};

% Draw the line Z' = rho*X'/sqrt(1-rho^2) constrained within axis bounds
\pgfmathsetmacro{\leftX}{-3/\slope}  % X value where line reaches Z' = -3
\pgfmathsetmacro{\rightX}{3/\slope}  % X value where line reaches Z' = 3
\pgfmathsetmacro{\actualLeftX}{max(-5.5, \leftX)}
\pgfmathsetmacro{\actualRightX}{min(5.5, \rightX)}
\pgfmathsetmacro{\actualLeftZ}{\slope*\actualLeftX}
\pgfmathsetmacro{\actualRightZ}{\slope*\actualRightX}
\draw[thick] (\actualLeftX, \actualLeftZ) -- (\actualRightX, \actualRightZ) node[above right] {$Z' = \frac{\rho X'}{\sqrt{1-\rho^2}}$};
%\draw[thick] (\actualLeftX, \actualLeftZ) -- (\actualRightX, \actualRightZ) 
%  node[above right, xshift=-20pt, yshift=-10pt] {$Z' = \frac{\rho X'}{\sqrt{1-\rho^2}}$};

% Coordinates for angle measurement
\coordinate (O) at (0,0);
\coordinate (Zaxis) at (0,1);
\coordinate (LinePoint) at (1, \slope);
\coordinate (NegZaxis) at (0,-1);
\coordinate (NegLinePoint) at (-1, -\slope);
\coordinate (Zfar) at (0,3.2);
\coordinate (linetopfar) at (3.2/\slope, 3.2);
\coordinate (NegZfar) at (0,-3.2);
\coordinate (Neglinetopfar) at (-3.2/\slope, -3.2);

    \fill[pattern=north west lines, pattern color=black!40] 
      (O) -- (NegZfar) -- (Neglinetopfar) -- cycle;

    \fill[pattern=north west lines, pattern color=black!40] 
      (O) -- (Zfar) -- (linetopfar) -- cycle;

% Draw angle arc between Z' axis and the line
\pic[
    draw,
    angle radius=20pt,
    angle eccentricity=0.6,
    preaction={draw=white, line width=6pt},
    "$\theta$"
] {angle = LinePoint--O--Zaxis};

\pic[
    draw,
    angle radius=20pt,
    angle eccentricity=0.6,
    preaction={draw=white, line width=6pt},
    "$\theta$"
] {angle = NegLinePoint--O--NegZaxis};

\end{tikzpicture}
        \captionof{figure}{$\{rX'^{2} < X'Z' \}$}
        \label{fig:proof_fig}
    \end{minipage}
\end{figure}
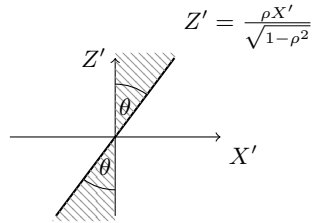

Note that $\theta = \text{arccos}(\rho)$. The distribution of $\binom{X'}{Z'}$ is rotationally invariant, so for any positive integer $m$, if $\frac{m\pi}{\theta} \in \mathbb{Z}$, we can copy and rotate the shaded area $\frac{m\pi}{\theta}$ times to cover the plane $m$ times (one covering of the plane corresponds to $\prob{X' \in \mathbb{R}, Z' \in \mathbb{R}} = 1$). Thus the shaded area is $\prob{X'Y'<0} = \frac{\theta}{\pi}$. Otherwise let $\theta_{m} = \frac{m\pi}{\left\lceil\frac{m\pi}{\theta} \right\rceil}$ and our argument applies if the shaded area were shrunk to $\theta_{m}$. Then $\theta_{m}$ tends to $\theta$ from below as $m \to \infty$ and by countable additivity of measures, $\prob{X'Y' < 0} = \lim_{m\to\infty} \frac{\theta_{m}}{\pi} = \frac{\theta}{\pi}$.
\end{proof}

We now characterise our losses for a class of linear $f$ which we interpret as a linearized GCN (see Sec \ref{sec:experimental_setup} for a specific example):
    
\begin{theorem}
\label{thm:general_class_err}
    \bs{If} the function $f$ generating the labels from the features can be written as $$f(\matr{X}) = \matr{G}\matr{X}\matr{w}$$ where $\matr{G}\in \mathbb{R}^{N \times N}$ and $\vect{w} \in \mathbb{R}^{d} \backslash \{\vect{0}\}$, then for a sample set $\set{S}$,
    \begin{align}
         &\text{Classification Loss}  = \sum_{i \in \set{V}} \frac{1}{\pi}\text{arccos}\left( \rho_{i}(\matr{G}) \right) \label{eq:gen_class_err}
        \end{align}
and if $d=1$,
        \begin{align}
        &\text{Reconstruction Loss}  =   \sum_{i \in \set{V}} \left(\sigma_{i}(\matr{I})\right)^{2} + \left(\nu_{i}(\matr{I})\right)^{2} - 2c_{i}(\matr{I}) \label{eq:gen_rec_err}
        %\sum_{i \in \set{V}} \frac{1}{\pi}\text{arccos}\left( \frac{\left(\matr{R}_{\set{S}}\matrsub[S,N]{\Sigma}\right)_{ii}}{\sqrt{\matr{\Sigma}_{ii}\left(\matr{R}_{\set{S}}\left(\matr{\Sigma} + n^{2}\matr{I}_{N}\right) \matr{R}_{\set{S}}^{T} \right)_{ii}}} \right)
    \end{align}
    where for any matrix $\matr{M}$, we let
    \begin{align}
        \matr{C} &= \text{Cov}(\matr{X}\vect{w}) \cdot (\sqnormvec{\vect{w}})^{-1},\\
        c_{i}(\matr{M}) &={(\matr{M}\matr{R}_{\set{S}}\matrsub[S,N]{C} \matr{M}^{T})_{ii}}, \\
        (\sigma_{i}(\matr{M}))^{2} &= (\matr{M}\matr{C}\matr{M}^{T})_{ii}, \\
        (\nu_{i}(\matr{M}))^{2} &= ( \matr{M}\matr{R}_{\set{S}}\msubgen[\set{S}]{\matr{C} + \eta^{2}\matr{I}_{N}} \matr{R}_{\set{S}}^{T} \matr{M}^{T})_{ii},\\
        \rho_{i}\left(\matr{M}\right) &= \frac{c_{i}(\matr{M})}{\sigma_{i}(\matr{M})\nu_{i}(\matr{M})}.
    \end{align}
\end{theorem}
\begin{proof}
    %For each node, $\expect{f(\matr{X})} = \mathbb{E}[{f(\hat{\matr{X}})}] = 0$ so 
    As $\matr{X}$ is independent to the noise and $\expect{\matr{X}}=0$,
    $\text{Cov}(f(\hat{\vect{X}})_{i}, f({\vect{X}})_{i}) = \expect{f(\hat{\vect{X}})f({\vect{X}})^{T}}_{ii} = \matr{G}\matr{R}_{\set{S}}\matrsub[S,N]{I}\expect{(\matr{X} + \eta \cdot \vect{\epsilon})\vect{w}\vect{w}^{T}\matr{X}^{T}}\matr{G}^{T}$ = $\matr{G}\matr{R}_{\set{S}}\msubgen[\set{S},\set{N}]{\text{Cov}(\matr{X}\vect{w})}\matr{G}^{T} = \sqnormvec{\vect{w}}c_{i}(\matr{G})$. Similarly, $\text{Var}(f(\matr{X})_{i}) = \sqnormvec{\vect{w}}(\sigma_{i}(\matr{G}))^{2}$ and  $\text{Var}(f(\hat{\matr{X}})_{i}) = \sqnormvec{\vect{w}}(\nu_{i}(\matr{G}))^{2}$ giving $\rho_{i}(\matr{G}) = \text{Corr}(f(\hat{\vect{X}})_{i}, f({\vect{X}})_{i})$. For classification loss, as $(f(\hat{\vect{X}})_{i}$ and $f({\vect{X}})_{i})$ are linear transformations of $\matr{X}$, they are jointly Gaussian and we can apply Lemma \ref{lemma:gaussian_sign} to get that the misclassification probability at node $i$ is $\frac{\text{arccos}(\rho_{i}(\matr{G}))}{\pi}$; sum across all nodes to get the total loss. For the Reconstruction loss, $\expect{||\matr{X} - \hat{\matr{X}}||^{2}_{F}} =  \expect{\text{tr}(\matr{X}\matr{X}^{T})} + \expect{\text{tr}(\hat{\matr{X}}\hat{\matr{X}}^{T})} - 2\expect{\text{tr}(\hat{\matr{X}}\matr{X}^{T})}  = \sum_{i \in \set{V}}d \text{Var}(\matr{X}_{i}) + d\text{Var}(\hat{\matr{X}}_{i}) - 2d\text{Cov}(\hat{\matr{X}}_{i}, \matr{X}_{i}) $. As $\expect{\matr{X}} = \expect{\hat{\matr{X}}} = 0$, if $d=1$ this expands to (\ref{eq:gen_rec_err}).
\end{proof}

% Note that classification loss is only dependent on $\matr{G}$, not $\vect{w}$.

\bs{
If we let $\matr{G}$ be a polynomial $p$ of $\tilde{\matr{A}}_{\gamma}$, then $\vect{l} = \text{sign}(p(\tilde{\matr{A}}_{\gamma})\matr{X}\matr{w})$, which corresponds to a multi-layer SGC \cite{wu2019simplifying} or a single layer polynomial-based GNN\footnote{We assume that all feature dimensions share the same polynomial filter.} like Chebnet \cite{defferrard2016convolutional}, FavardGNN \cite{guo2023Favard}, JacobiNet \cite{guo2023manipulating} or BernNet \cite{he2021bernnet}. Then the classification loss in Theorem \ref{thm:general_class_err} is the difference in output when applying the GNN to the full data $\matr{X}$ versus to the reconstructed data $\hat{\matr{X}}$. 
Note that in the case where the columns of $\matr{X}$ are independent, $\matr{C} = \matr{\Sigma}$, and therefore for SGC the classification loss does not depend on the weights, only on the number of layers.
}

%We can use \bs{Theorem \ref{thm:general_class_err}} to calculate the classification \bs{loss} for an $r$-layer \bs{linearized} GCN without bias, setting $\matr{G} = (\tilde{\matr{A}}_{\gamma})^{r}$ and $\vect{w} = \prod \matr{W}_{i}$. The classification \bs{loss then} does not depend on the weights \bs{$\vect{w}$} -- if our linearized GCN is an accurate representation of the true function, then the classification \bs{loss} is not dependent on its weights, only its structure \bs{$\matr{G}$} (i.e. the number of layers). 

\bs{ To interpret Theorem \ref{thm:general_class_err}, we consider any $\matr{G}$ and LS reconstruction of a noise-free $k$-bandlimited signal; specifically $\matr{\Sigma} = \projbl$. In this case, $\sqrt{c_{i}(\matr{G})} = \nu_{i}(\matr{G})$ and the {\text{classification loss}} is $\sum_{i \in \set{V}} \frac{1}{\pi}\arccos{\left(\sqrt{\frac{c_{i}}{\sigma_{i}}}\right)}$ where $0 \leq \frac{c_{i}}{\sigma_{i}} \leq 1$. Thus, we expect the behaviour of classification \bs{loss} as it changes with sample size to look similar to $x \to \arccos{\left(\sqrt{x}\right)}$.

%We briefly contrast our result to the literature on moderation. Mackay \cite{mackay1992evidence} noted that uncertainty -- lack of information -- over the weights causes a sigmoid shaped posterior probability. We note that in the above case, lack of information about the data causes a sigmoid shaped error as sample size varies.
}

\subsection{Relationship between Reconstruction and Classification Loss}
\bs{
We now relate the reconstruction and classification losses via the error in reconstructing the output of $f$. All statements in the following Corollary hold for any integer value of $d \geq 1$.

\begin{corollary}
\label{corr:relationship}
    Assume the setting and definitions of Theorem \ref{thm:general_class_err}. 
    Define the normalized output error at node $i$ as 
    \begin{equation}
    \text{Error}_{out, i} = \left(\sqnormvec{\vect{w}}\right)^{-1}{\expect{||(f(\matr{X}))_i - (f(\hat{\matr{X}}))_i||^{2}_{2}}}. 
    \end{equation}
%Then if for some $r \geq 0, \gamma \geq 0$, $\matr{G} =  \tilde{\matr{A}}_{\gamma}^{r}$,
\bs{Let $||\matr{G}||$ be the spectral norm of $\matr{G}$, then }
\begin{equation}
\label{eq:bound_final_layer}
\sum_{i \in \set{V}} \text{Error}_{out,i} \leq ||\matr{G}||^{2} \cdot \text{Reconstruction Loss}
\end{equation}
and for every node $i$, the following is always a valid triangle.
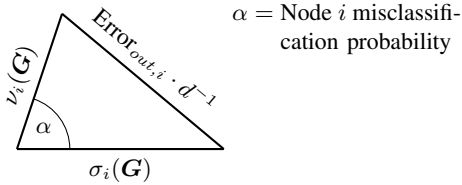
\begin{figure}[H]
    \centering
    \begin{tikzpicture}[scale=0.5]
    \def\slope{3}
    \def\xA{5.5}
    \def\yA{0}
    \def\xB{1.2}
    \def\yB{\xB * \slope}
    \pgfmathsetmacro{\etaValue}{atan(\slope)}
    
    %\pgfmathsetmacro{\slope}{\rhovalue/sqrt(1-\rhovalue^2)}
    % Define the vertices of the triangle
    \coordinate (O) at (0,0);
    \coordinate (A) at (\xA,\yA);
    \coordinate (B) at (\xB,\yB);
    
    \pgfmathsetmacro{\otherRot}{atan2(\yA - \yB, \xA - \xB}

    % Draw the triangle
    %\draw[thick] (O) -- (A) -- (B) -- cycle;
    
    \draw[thick] (O)--(A)  node[midway, below] {$\sigma_{i}(\matr{G})$}; %{$\sqrt{\text{Var}(f(\matr{X}))_i}$};
    
    \draw[thick] (O)--(B)  node[midway, above, rotate=\etaValue] {$\nu_{i}(\matr{G})$}; %{$\sqrt{\text{Var}(f(\hat{\matr{X}}))_i}$};
    
    \draw[thick] (A)--(B)  node[midway, above, rotate=\otherRot, align=center] {
    $\text{Error}_{out,i} \cdot d^{-1} $
    };

    %\draw[thick] (A)--(B)  node[midway, above, rotate=\otherRot, align=center, text width=100pt] {Final layer reconstruction error at node $i$}; %{$\sqrt{\expect{||(f(\matr{X}))_i - (f(\hat{\matr{X}}))_i||^2}}$} ;
    % Label the vertices
    \iffalse
    \node[below left] at (O) {$0$};
    \node[below right] at (A) {$f(\matr{X})_i$};
    \node[above] at (B) {$f(\hat{\matr{X}})_i$};
    \fi

\pic[
    draw,
    angle radius=20pt,
    angle eccentricity=1,
    preaction={draw=white, line width=6pt},
    %"Classification\\Error",%"${\eta}$"
] {angle = A--O--B};
\node[right, xshift=4pt, yshift=8pt, text width=50pt] at (O) {$\alpha$};%{Node $i$ Classification Error };
%\node[right, rotate=\etaValue/2, xshift=20pt, yshift=0pt, text width=50pt] at (O) {Node $i$ Misclassification Probability};%{Node $i$ Classification Error };

\node[below, right] at (\xA,\yB-0.4) {$\alpha = \parbox[t][][t]{2.4cm}{Node $i$ misclassification probability }$};
%\node[above, text width=3cm] at (\xA,-2) {Node $i$ Misclassification Probability};

    % Add small circles at vertices for clarity
    \iffalse
    \fill (O) circle (1pt);
    \fill (A) circle (1pt);
    \fill (B) circle (1pt);
    \fi
\end{tikzpicture}
    \caption{Relationship between \bs{misclassification probability} \iffalse Classification \bs{Loss} \fi (angle) and \bs{normalized output error} at node $i$}
    \label{fig:class_rec_rel_fig}
\end{figure}

\end{corollary}
\begin{proof}
    We first prove \eqref{eq:bound_final_layer}. %Write the spectral norm of a matrix $\matr{M}$ as  $||\matr{M}||$. By Gershgorin's Circle Theorem, $||\tilde{\matr{A}}|| \leq 1$; thus by \cite[Theorem 1]{wu2019simplifying}, $||\tilde{\matr{A}}_{\gamma}|| \leq 1$. Therefore $||\matr{G}|| \leq 1$. 
    By linearity of $f$, $(f(\matr{X}))_i - (f(\hat{\matr{X}}))_i = (\matr{G}(\matr{X} - \hat{\matr{X}})\vect{w})_{i}$. Then  $\sum_{i \in \set{V}} ||f(\matr{X})_{i} - f(\hat{\matr{X}})_{i} ||_{2}^{2} = ||\matr{G}(\matr{X} - \hat{\matr{X}})\vect{w}||_{2}^{2} \leq ||\matr{G}||^{2}\cdot ||(\matr{X} - \hat{\matr{X}})\vect{w}||_{2}^{2} \leq ||\matr{G}||^{2} \cdot (\matr{X} - \hat{\matr{X}})||_{F}^{2} \cdot ||\vect{w}||_{2}^{2}$ where the last inequality is because $||\cdot||_{2} = ||\cdot||_{F}$ for vectors and from sub-multiplicativity of $||\cdot||_{F}$. Taking expectations on both sides gives $ \sqnormvec{\vect{w}} \cdot \sum_{i \in \set{V}} \text{Error}_{out,i} \leq ||\matr{G}||^{2} \cdot  \text{Reconstruction Loss} \cdot ||\vect{w}||_{2}^{2}$. 

    Secondly, we prove Fig. \ref{fig:class_rec_rel_fig} forms a valid triangle. Through the same expansion and independence arguments as in the proof of Theorem \ref{thm:general_class_err}, $\text{Error}_{out,i} = d\cdot \left((\sigma_{i}(\matr{G}))^{2} + (\nu_{i}(\matr{G}))^{2} - 2\rho_{i}(\matr{G}) \right)$. As $\rho_{i}(\matr{G})$ is the probability of misclassification of node $i$ (see the proof of Theorem \ref{thm:general_class_err}), the sides and angle of Fig. \ref{fig:class_rec_rel_fig} obey the law of cosines and thus form a valid triangle.
\end{proof}
}

\bs{For SGCs, $||\matr{G}|| \leq 1$ (as $||\tilde{\matr{A}}|| \leq 1$ by Gershgorin's Disc Theorem \cite{horn-and-johnson}, so $\forall \gamma\geq 0:||\tilde{\matr{A}}_{\gamma}|| \leq 1$ by \cite[Theorem 1]{wu2019simplifying} and therefore $||\tilde{\matr{A}}_{\gamma}^{r}|| \leq 1$ for all $r \geq 1$). Therefore, by \eqref{eq:bound_final_layer}, the sum of the per-node normalized output errors (which are directly related to the per-node misclassification probabilities, as shown in Fig. \ref{fig:class_rec_rel_fig}) is upper bounded by the reconstruction loss.}

While inequality \eqref{eq:bound_final_layer} suggests optimising reconstruction loss to minimise $\sum_{i \in \set{V}} \text{Error}_{out,i}$, the non-linear relationship between $\text{Error}_{out,i}$ and the classification loss at node $i$ make this sum a poor proxy for total classification error. We discuss experimental results around optimising reconstruction loss as a proxy for classification loss in Section \ref{sec:optimal_sampling_results}. % , thus r % the non-linear

\bs{
\subsection{Sampling for Classification}
\label{sec:optimal_sampling_design}
We use Theorem \ref{thm:general_class_err} to design a sampling scheme specifically for classification, minimising mean classification loss rather than the common mean reconstruction loss objective \cite{wang2018optimal,wang2019low,sripathmanathan2024impact, mfn}. We propose greedily choosing $\set{S}$ to optimize \eqref{eq:gen_class_err} from Theorem \ref{thm:general_class_err}. %and compare it to random and A-optimal sampling under both reconstruction and classification.
}

\section{Empirical Results}
%We now validate our theoretical results by empirical experiments \bs{and compare the performance of our classification-optimal sampling scheme to other sample schemes designed with reconstruction in mind.} %and  We also apply our theoretical results to propose an optimal sampling scheme for classification under our setup, and compare its performance to other optimal sampling schemes aiming for reconstruction only.
\iffalse
BA:
[SGC, Simple] 
x
[LS,
%(GLR),
Feat Prop]

%[noiseless, noisy] x 
noisy vs noiseless for LS
[BA, SBM] - 1 figure SBM

%[bandlimited signal, L+ signal]
For each comparison each fig should have the same x \& y

classic setting :LS x bandlimited signal -- best performance for our sample selection scheme

%SGC: BL (1 layer: sign((AXW2)W2) )
%SGC: L+ (2 layersL
\fi
\subsection{Experimental Setup}
\label{sec:experimental_setup}
%We first present the setup of our experiments. All experiments are with regards to the symmetrically normalized Laplacian and its eigenbasis.

\subsubsection{Synthetic Graph Generation}
We generate 32 graphs with $N=500$ nodes from each of the following models:
\begin{itemize}
    \item Barab\'asi-Albert (BA) with a preferential attachment to 3 vertices at each step of construction;
    \item Stochastic Blockmodel (SBM) with intra- and inter- cluster edge probabilities of 0.7 and 0.1, respectively.
\end{itemize}

\iffalse
\subsubsection{Choice of $f$}
We consider two types of classification function. In the `simple' case, we consider $d=1$ and $f(x) = x$, so the label is the sign of the true feature. In the `sgc' case, we consider a random linearized GCN with added self-loops as a model for $f$. We choose our architecture to mimic the strong baselines in \cite{luo2024classic}. Specifically:
\begin{equation}
    f_{\text{sgc}, r}(\matr{X}) = (\matr{I} + \tilde{\matr{A}})^{r}\matr{X}\prod_{i=1}^{r+1}\matr{W}_{i}.
\end{equation}
Our first layer has $\matr{W}_{1}\in\mathbb{R}^{64\times 32}$, any intermediate layer has $\matr{W}_{i} \in \mathbb{R}^{32 \times 32}$ and our final readout layer has $\matr{W}_{r+1} \in \mathbb{R}^{32 \times 1}$. We populate $\matr{W}_{i}$ via Glorot initialization \cite{glorot2010understanding}. We choose $r=1$ to avoid label oversmoothing.
\fi

\subsubsection{Synthetic Signal Generation}
\bs{We draw features $\matr{X}$ as described in Section \ref{sec:problem_setting} with $\matr{\Sigma} = \projbl$, i.i.d. columns, bandwidth $k = \frac{N}{10}$ and noise variance $\eta^{2}=10^{-3}$ so the SNR is $20dB$, like in \cite{bai2020fast}. }
%We consider signals drawn from $\mathcal{N}(\vect{0},\projbl)$ (i.e. $\matr{\Sigma} = \projbl$) with bandwidth $k = \frac{N}{10}$. We consider both noiseless signals and signals with $\mathcal{N}(\vect{0}, 10^{-3} \matr{I}_{N})$ noise added ($n=10^{-1.5}$) so the SNR is 20dB, like in \cite{bai2020fast}. \bs{We draw a single feature matrix $\matr{X}\in\mathbb{R}^{500 \times 64}$ per graph.}
%We draw 200 signals per graph instantiation in the `simple' $f$ case (i.e. $32 \times 200$ classification tasks per graph model). We draw one task per graph in the `sgc' $f$ case (corresponding to $\matr{X} \in \mathbb{R}^{500 \times 64}$).

\subsubsection{Choice of $f$}
\bs{For each graph, we construct $f$ as a linearized GCN with randomly initialized weights and added self-loops ($\gamma=1)$. We choose our architecture to mimic the strong baselines in \cite{luo2024classic}, but keep to a single convolutional layer to minimise oversmoothing.
Specifically, we set
\begin{equation}
 f(\matr{X}) = \tilde{\matr{A}}_{1}\matr{X}\left(\matr{W}_{1}\matr{W}_{2} \right)
\end{equation}
 where $\matr{W}_{1} \in \mathbb{R}^{64 \times 32}$, $\matr{W}_{2} \in \mathbb{R}^{32 \times 1}$ are both populated via Glorot initialization \cite{glorot2010understanding}.}

\subsubsection{Real-world Datasets}
We consider an FMRI dataset with nodes corresponding to brain regions and 292 graph signals corresponding to blood oxygen levels \cite{zhi2023gaussian, sripathmanathan2024impact}, subsampling a connected graph with $N=367$ nodes using neighbourhood sampling \cite{hamilton2017inductive}. We do not bandlimit the \bs{given} signals. We construct a classification task by subtracting the mean from all signals and taking the sign as node class labels, representing high or low blood oxygen levels.

%We use optimal sampling assuming an SNR of 20dB (i.e., $n=10^{-1.5}$) and reconstruct with a bandlimit of $\lfloor \frac{N}{10} \rfloor$.

\subsubsection{Sample Set Selection}
We consider three sampling methods:
\begin{itemize}
    \item Random sampling;
    \item Greedy sampling optimizing classification \bs{loss} \eqref{eq:gen_class_err}, which is described in Sec \ref{sec:optimal_sampling_design};
    \item Greedy sampling optimizing reconstruction \bs{loss} \eqref{eq:gen_rec_err}. For LS, this is the same as optimizing for the frequently used A-optimal objective for reconstruction \cite{wang2018optimal, wang2019low, sripathmanathan2024impact, mfn}.
\end{itemize}

%Under LS reconstruction, \iffalse and $\matr{\Sigma}=\projbl$\fi multiple nodes are frequently tied in reconstruction error. In this case, we also optimize for noise sensitivity, and this sample choice corresponds to A-optimal sampling \cite{wang2018optimal}.

%\bs{In our setup, greedily optimizing reconstruction loss is the same as greedily optimizing the frequently used A-optimal objective for reconstruction \cite{wang2018optimal, wang2019low, sripathmanathan2024impact, mfn}.}

\bs{To sample on real world data, we need to compute losses for our sampling to optimize. This requires making assumptions about the noise level and distribution of the features. For the FMRI dataset, we assume the features are i.i.d. $\mathcal{N}(\vect{0},\projbl)$ with $k = \lfloor \frac{N}{10}\rfloor$ and that the SNR is $20dB$ (i.e., $\eta^{2}=10^{-3}$). }

\setlength{\belowcaptionskip}{-3pt}
%\captionsetup[subfigure]{aboveskip=-1pt,belowskip=-4pt}
\captionsetup[subfigure]{aboveskip=-1pt,belowskip=-2.5pt}
%\setlength{\abovecaptionskip}{-3pt}
%\setlength{\aftercaptionskip}{-0.3cm}
%\setlength{\intextsep}{0pt plus 2pt}
\iffalse
%%% Simple, LS, BA, Noiseless
\begin{figure*}%
    \centering
    \begin{subfigure}[T]{0.5\columnwidth}
    \resizebox{\width}{0.62\columnwidth}{
    \includegraphics[width=\columnwidth]{plots/simple/BA_bl_LS_0_noise_500_nodes_class_vs_rec.png}}
    \caption{Loss Comparison \\}
    \label{subsubfig:BA_noiseless_simple_LS_comp}
    \end{subfigure}
    \hfill
    \begin{subfigure}[T]{0.5\columnwidth}
    \resizebox{\width}{0.62\columnwidth}{
    \includegraphics[width=\columnwidth]{plots/simple/BA_bl_LS_0_noise_500_nodes_recon.png}}
    \caption{Reconstruction Loss}%
    \label{subsubfig:BA_noiseless_simple_LS_rec}
    \end{subfigure}
    \hfill%
    \begin{subfigure}[T]{0.5\columnwidth}
    \resizebox{\width}{0.62\columnwidth}{
    \includegraphics[width=\columnwidth]{plots/simple/BA_bl_LS_0_noise_500_nodes_classif.png}}
    \caption{Class. Loss (Analytic)}
    \label{subsubfig:BA_noiseless_simple_LS_class_anal}
    \end{subfigure}%
    \hfill
    \begin{subfigure}[T]{0.5\columnwidth}
    \resizebox{\width}{0.62\columnwidth}{
    \includegraphics[width=\columnwidth]{plots/simple/BA_bl_LS_0_noise_500_nodes_class_empiric.png}}
    \caption{Class. Loss (Empirical)}%
    \label{subsubfig:BA_noiseless_simple_LS_class_emp}
    \end{subfigure}%
    \caption{LS Reconstruction, BA graph model, noiseless}
\label{subfig:BA_noiseless_simple_LS}
\end{figure*}
\fi
%%% SGC, LS, BA, Noiseless
\begin{figure*}%
    \centering
    \iffalse
    \begin{subfigure}[T]{0.5\columnwidth}
    \resizebox{\width}{0.62\columnwidth}{
    \includegraphics[width=\columnwidth, trim={1000cm 2cm 0 0}, clip]{plots/sgc-paper/BA_bl_LS_0_noise_500_nodes_1_layers_class_vs_rec.png}}
    \caption{Loss Comparison \\}
    \label{subsubfig:BA_noiseless_sgc_LS_comp}
    \end{subfigure}
    \hfill
    \fi
    \begin{subfigure}[T]{0.5\columnwidth}
    \resizebox{\width}{0.62\columnwidth}{
    \includegraphics[width=\columnwidth]{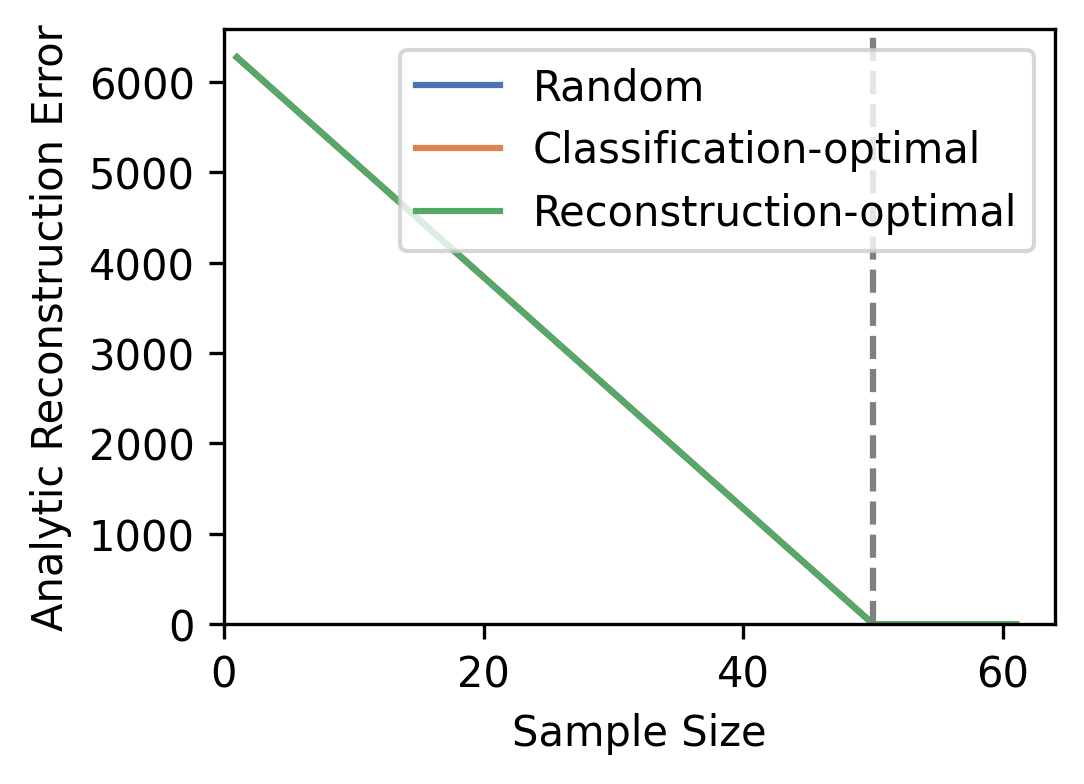}}
    \caption{Reconstruction Loss}%
    \label{subsubfig:BA_noiseless_sgc_LS_rec}
    \end{subfigure}
    %\hfill%
    \begin{subfigure}[T]{0.5\columnwidth}
    \resizebox{\width}{0.62\columnwidth}{
    \includegraphics[width=\columnwidth]{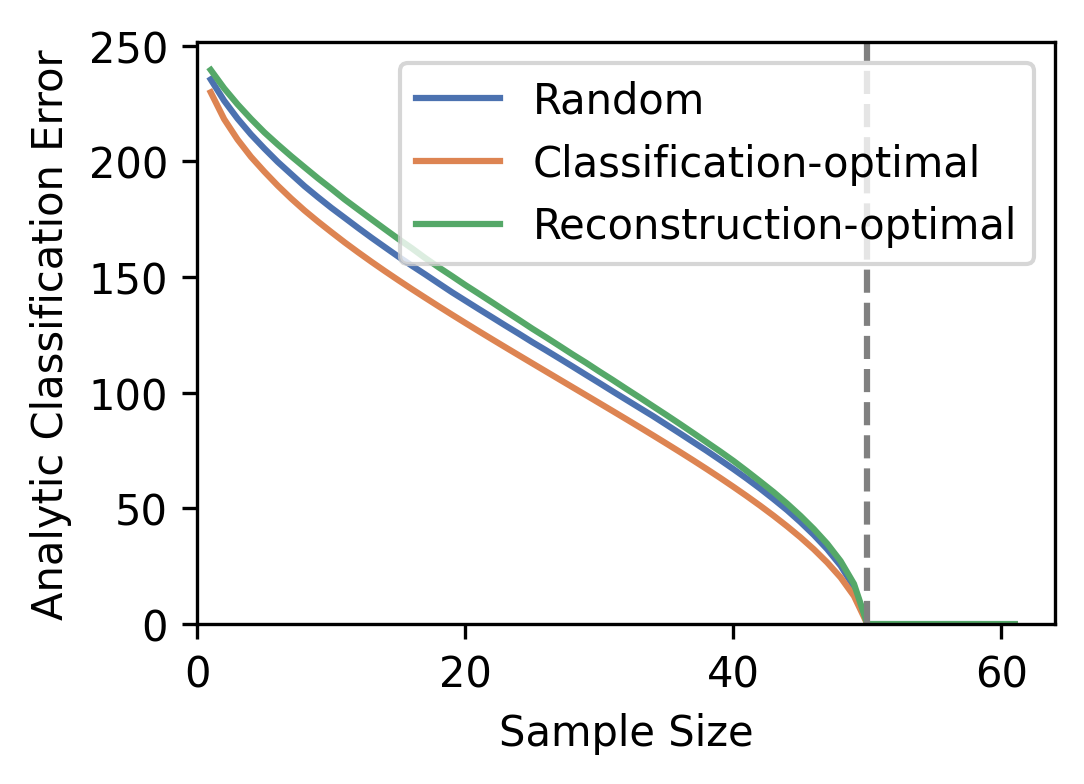}}
    \caption{Class. Loss (Analytic)}
    \label{subsubfig:BA_noiseless_sgc_LS_class_anal}
    \end{subfigure}%
    %\hfill
    \begin{subfigure}[T]{0.5\columnwidth}
    \resizebox{\width}{0.62\columnwidth}{
    \includegraphics[width=\columnwidth]{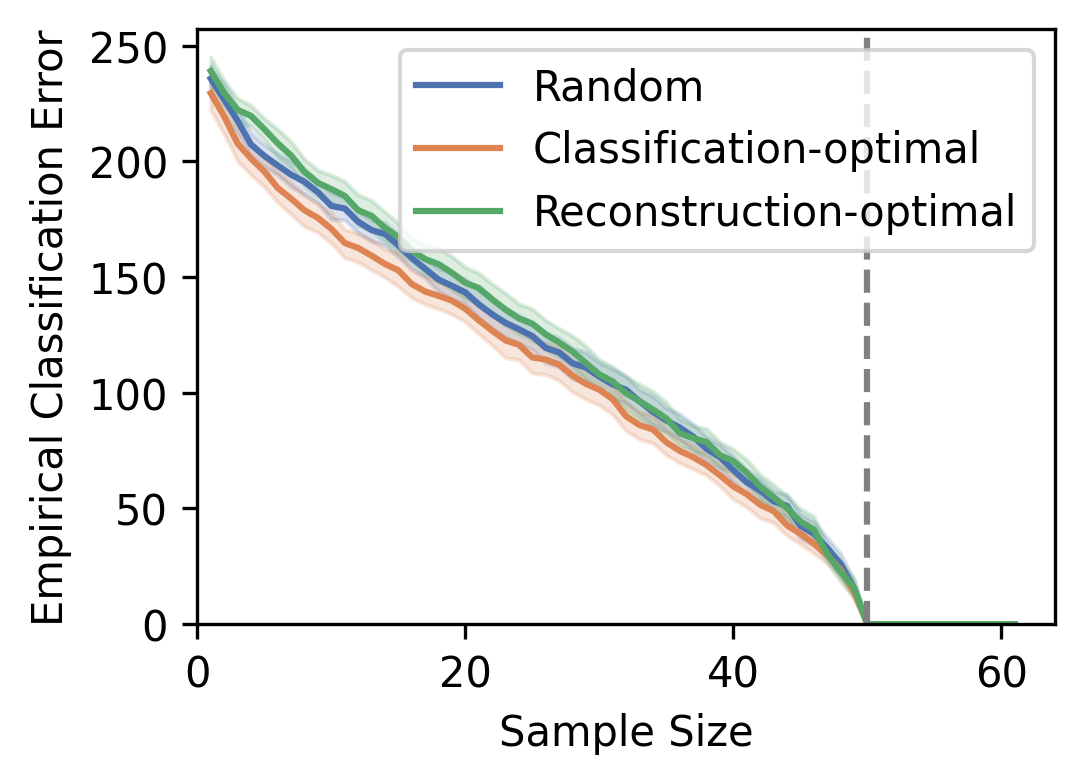}}
    \caption{Class. Loss (Empirical)}%
    \label{subsubfig:BA_noiseless_sgc_LS_class_emp}
    \end{subfigure}%
    \caption{LS Reconstruction, BA graph model, noiseless}
    \vspace{-0.19cm}
\label{subfig:BA_noiseless_sgc_LS}
\end{figure*}

%%% SGC, feat_prop, BA, Noiseless
\begin{figure*}%
    \centering
    \iffalse
    \begin{subfigure}[T]{0.5\columnwidth}
    \resizebox{\width}{0.62\columnwidth}{
    \includegraphics[width=\columnwidth]{plots/sgc-paper/BA_bl_feat_prop_0_noise_200_nodes_1_layers_class_vs_rec.png}}
    \caption{Loss Comparison \\}
    \label{subsubfig:BA_noiseless_sgc_feat_prop_comp}
    \end{subfigure}
    \hfill
    \fi
    \begin{subfigure}[T]{0.5\columnwidth}
    \resizebox{\width}{0.62\columnwidth}{
    \includegraphics[width=\columnwidth]{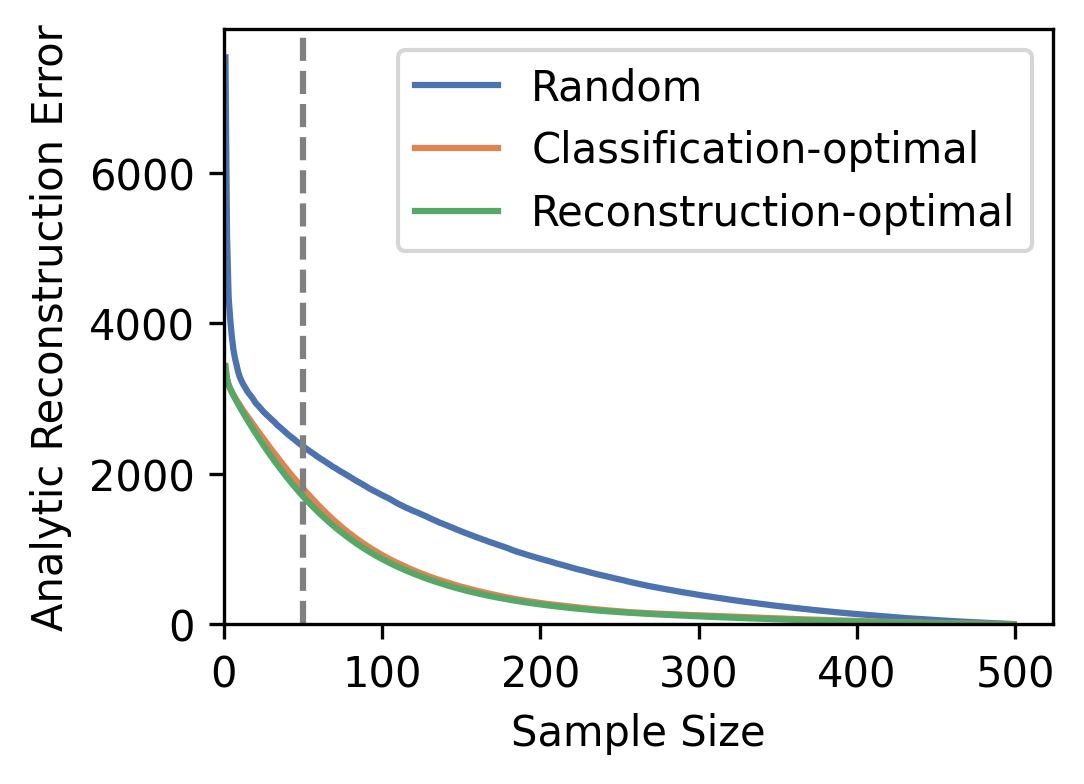}}
    \caption{Reconstruction Loss}%
    \label{subsubfig:BA_noiseless_sgc_feat_prop_rec}
    \end{subfigure}
    %\hfill%
    \begin{subfigure}[T]{0.5\columnwidth}
    \resizebox{\width}{0.62\columnwidth}{
    \includegraphics[width=\columnwidth]{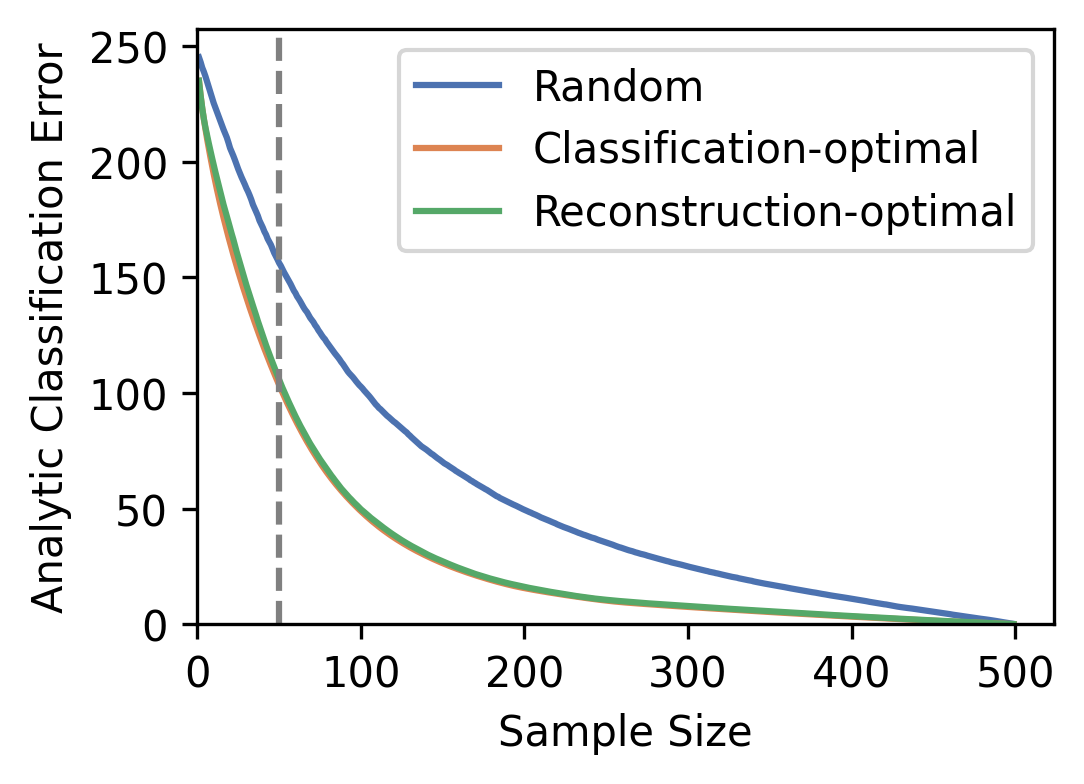}}
    \caption{Class. Loss (Analytic)}
    \label{subsubfig:BA_noiseless_sgc_feat_prop_class_anal}
    \end{subfigure}%
    %\hfill
    \begin{subfigure}[T]{0.5\columnwidth}
    \resizebox{\width}{0.62\columnwidth}{
    \includegraphics[width=\columnwidth]{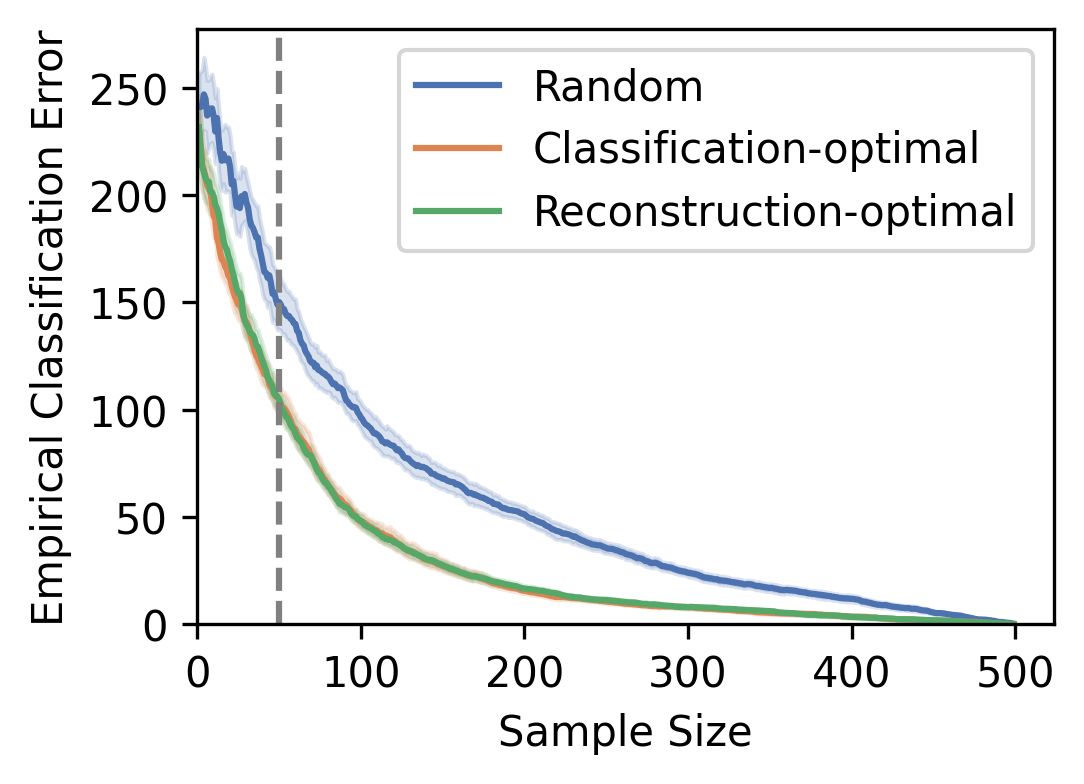}}
    \caption{Class. Loss (Empirical)}%
    \label{subsubfig:BA_noiseless_sgc_feat_prop_class_emp}
    \end{subfigure}%
    \caption{FP Reconstruction, BA graph model, noiseless}
    \vspace{-0.19cm}
\label{subfig:BA_noiseless_sgc_feat_prop}
\end{figure*}

%%% SGC, LS, BA, Noisy
\begin{figure*}%

    \centering
    \iffalse
    \begin{subfigure}[T]{0.5\columnwidth}
    \resizebox{\width}{0.62\columnwidth}{
    \includegraphics[width=\columnwidth]{plots/sgc-paper/BA_bl_LS_0.03162277660168379_noise_500_nodes_1_layers_class_vs_rec.png}}
    \caption{Loss Comparison}
    \label{subsubfig:BA_noisy_sgc_LS_comp}
    \end{subfigure}
    \hfill
    \fi
    \begin{subfigure}[T]{0.5\columnwidth}
    \resizebox{\width}{0.62\columnwidth}{
    \includegraphics[width=\columnwidth]{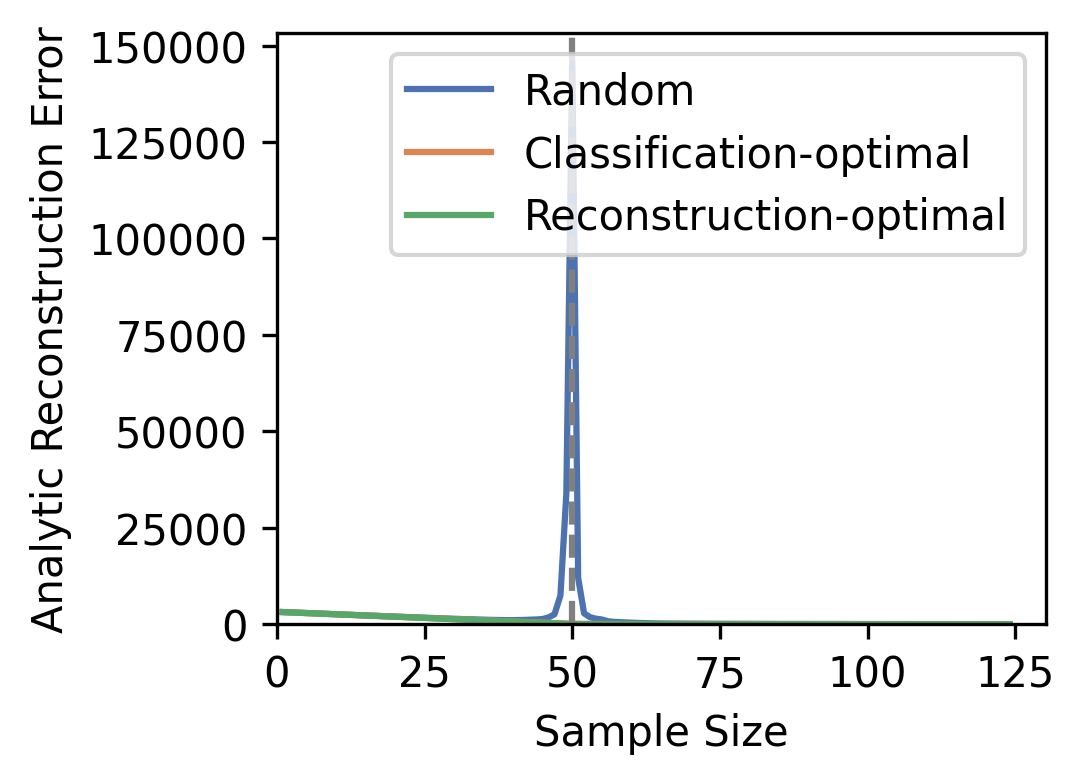}}
    \caption{Reconstruction Loss}%
    \label{subsubfig:BA_noisy_sgc_LS_rec}
    \end{subfigure}
    %\hfill%
    \begin{subfigure}[T]{0.5\columnwidth}
    \resizebox{\width}{0.62\columnwidth}{
    \includegraphics[width=\columnwidth]{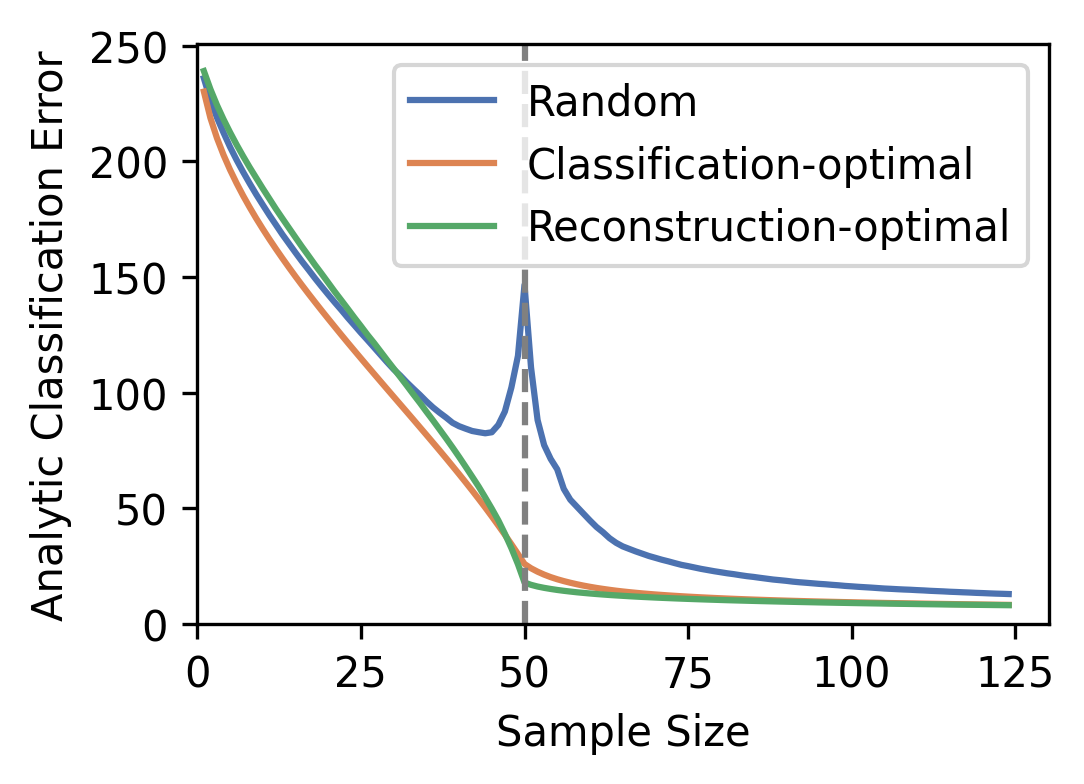}}
    \caption{Class. Loss (Analytic)}
    \label{subsubfig:BA_noisy_sgc_LS_class_anal}
    \end{subfigure}%
    %\hfill
    \begin{subfigure}[T]{0.5\columnwidth}
    \resizebox{\width}{0.62\columnwidth}{
    \includegraphics[width=\columnwidth]{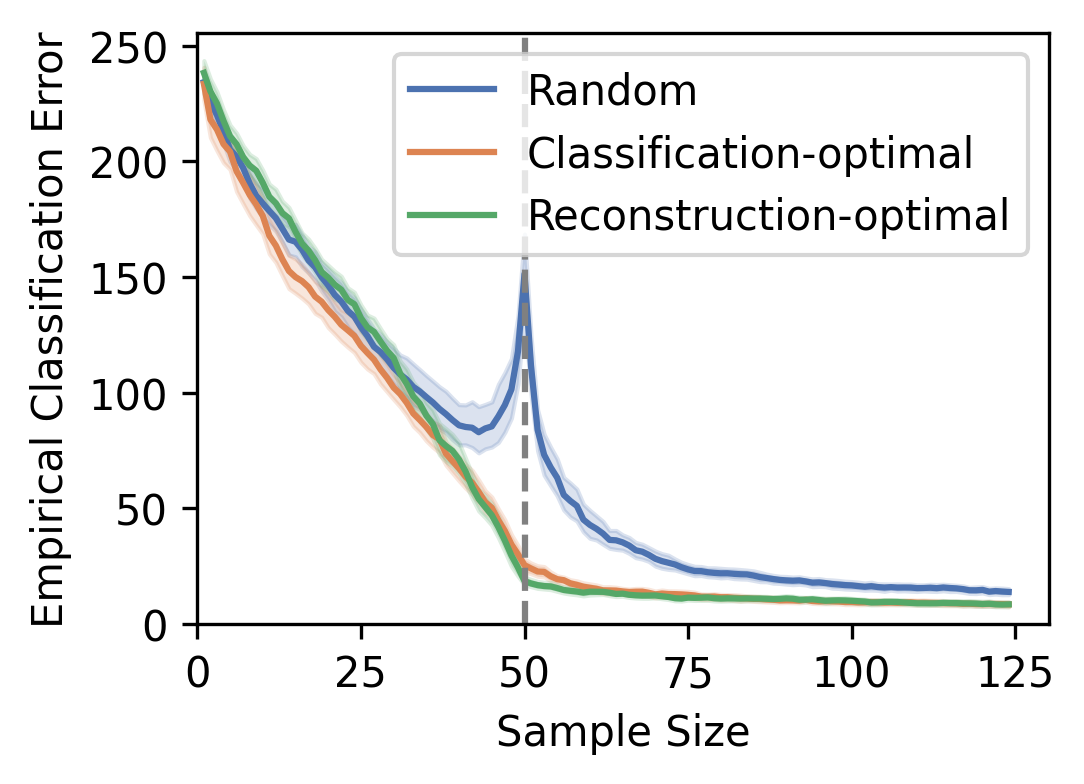}}
    \caption{Class. Loss (Empirical)}%
    \label{subsubfig:BA_noisy_sgc_LS_class_emp}
    \end{subfigure}%
    \caption{LS Reconstruction, BA graph model, Noisy ($20dB$ signal)}
    \vspace{-0.19cm}
\label{subfig:BA_noisy_vs_noisy_sgc_LS}
\end{figure*}

%%% SGC, LS, SBM, Noiseless
\begin{figure*}[h]%
    \centering
    \iffalse
    \begin{subfigure}[T]{0.5\columnwidth}
    \resizebox{\width}{0.62\columnwidth}{
    \includegraphics[width=\columnwidth]{plots/sgc-paper/SBM_bl_LS_0_noise_500_nodes_1_layers_class_vs_rec.png}}
    \caption{Loss Comparison \\}
    \label{subsubfig:SBM_noiseless_sgc_LS_comp}
    \end{subfigure}
    \fi
    %\hfill
    \begin{subfigure}[T]{0.5\columnwidth}
    \resizebox{\width}{0.62\columnwidth}{
    \includegraphics[width=\columnwidth]{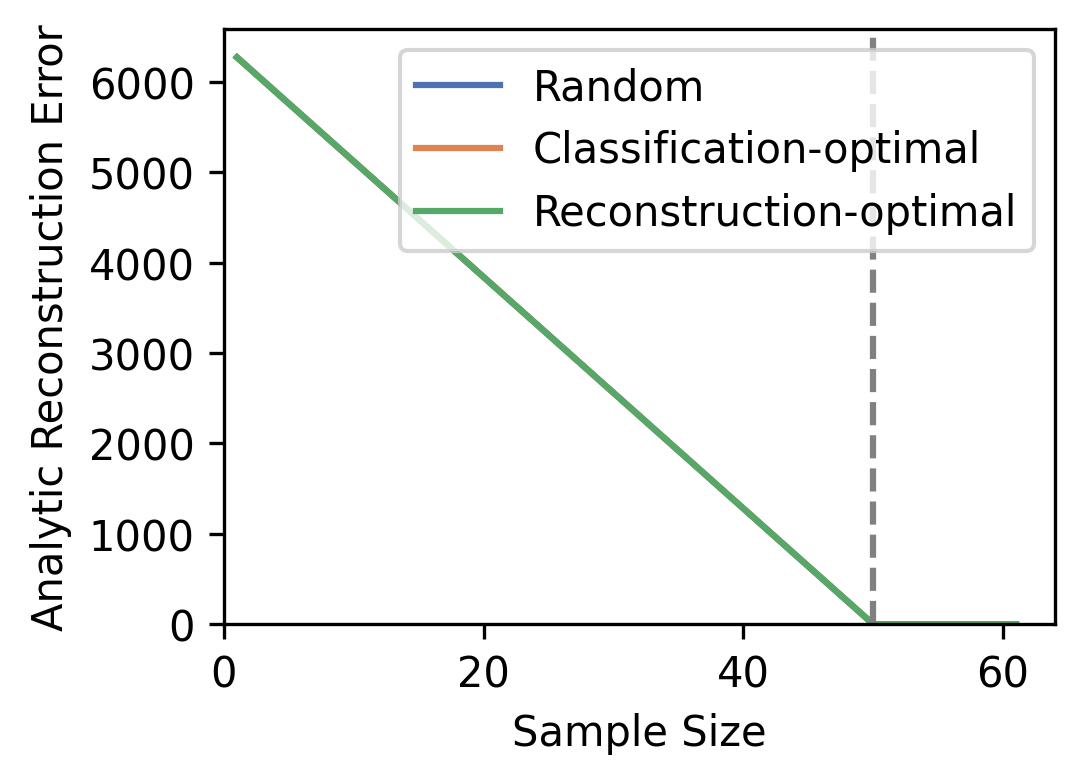}}
    \caption{Reconstruction Loss}%
    \label{subsubfig:SBM_noiseless_sgc_LS_rec}
    \end{subfigure}
    %\hfill%
    \begin{subfigure}[T]{0.5\columnwidth}
    \resizebox{\width}{0.62\columnwidth}{
    \includegraphics[width=\columnwidth]{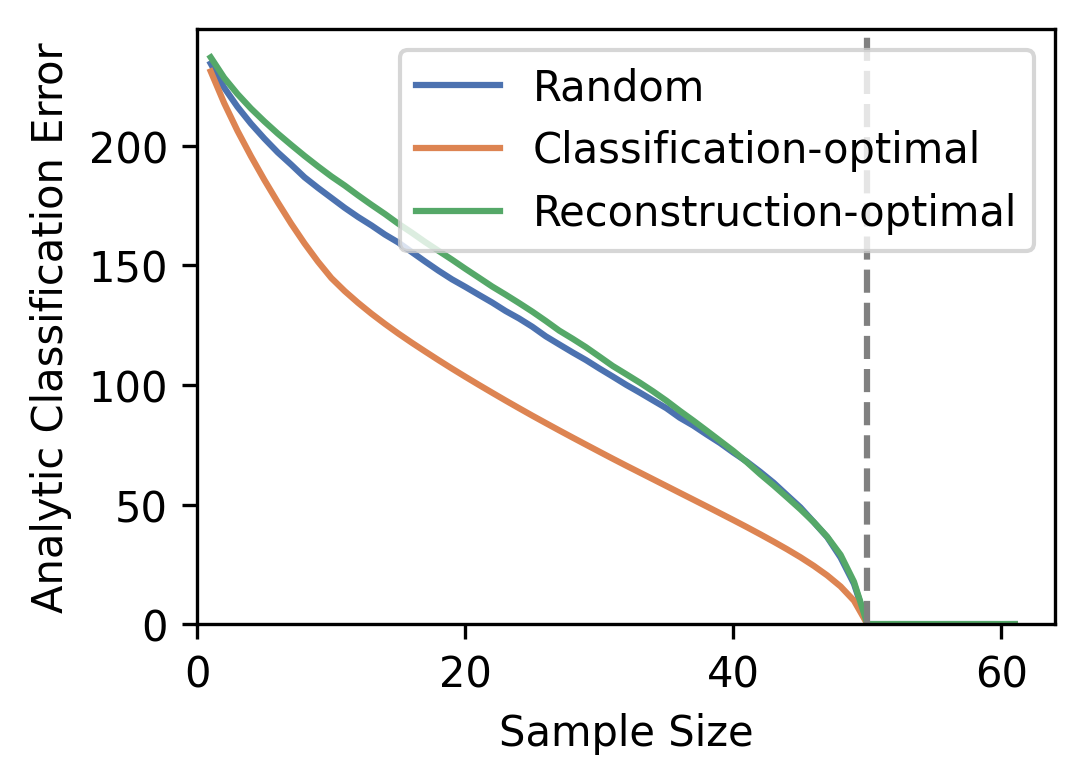}}
    \caption{Class. Loss (Analytic)}
    \label{subsubfig:SBM_noiseless_sgc_LS_class_anal}
    \end{subfigure}%
    %\hfill
    \begin{subfigure}[T]{0.5\columnwidth}
    \resizebox{\width}{0.62\columnwidth}{
    \includegraphics[width=\columnwidth]{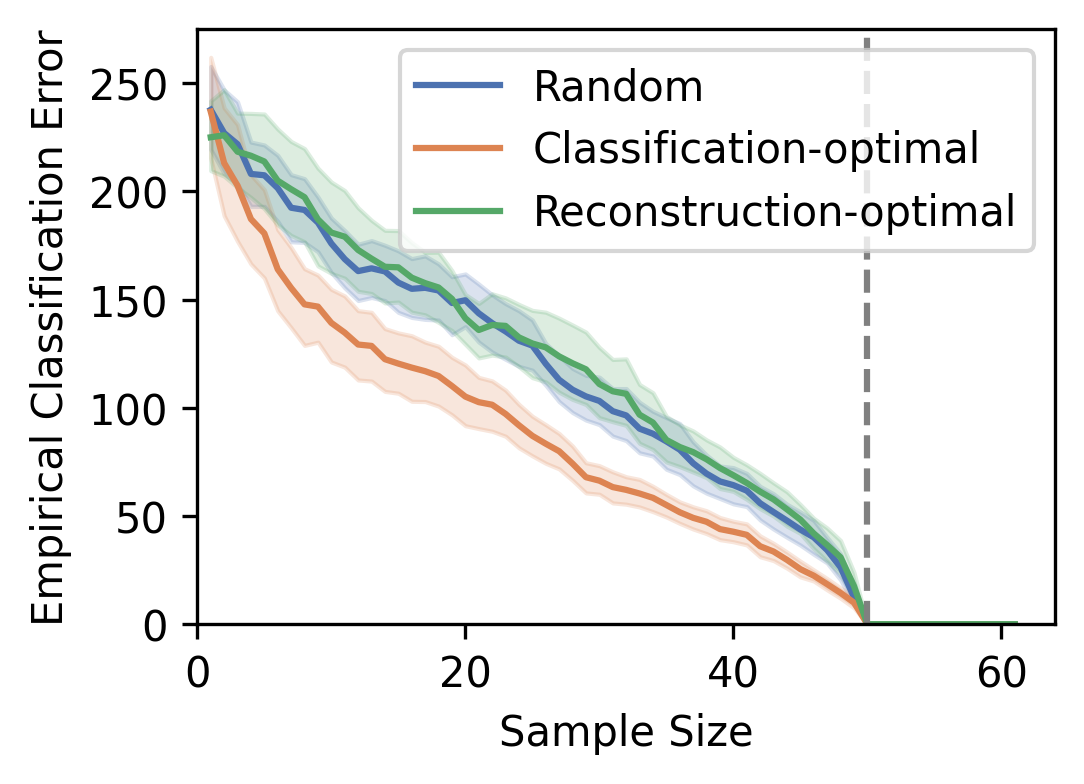}}
    \caption{Class. Loss (Empirical)}%
    \label{subsubfig:SBM_noiseless_sgc_LS_class_emp}
    \end{subfigure}%
    \caption{LS Reconstruction, SBM graph model, noiseless}
    \vspace{-0.19cm}
\label{subfig:SBM_noiseless_sgc_LS}
\end{figure*}

\begin{figure}[h!btp]%
    \begin{subfigure}[T]{0.45\columnwidth}
    \resizebox{\width}{0.62\columnwidth}{
    \includegraphics[width=\columnwidth]{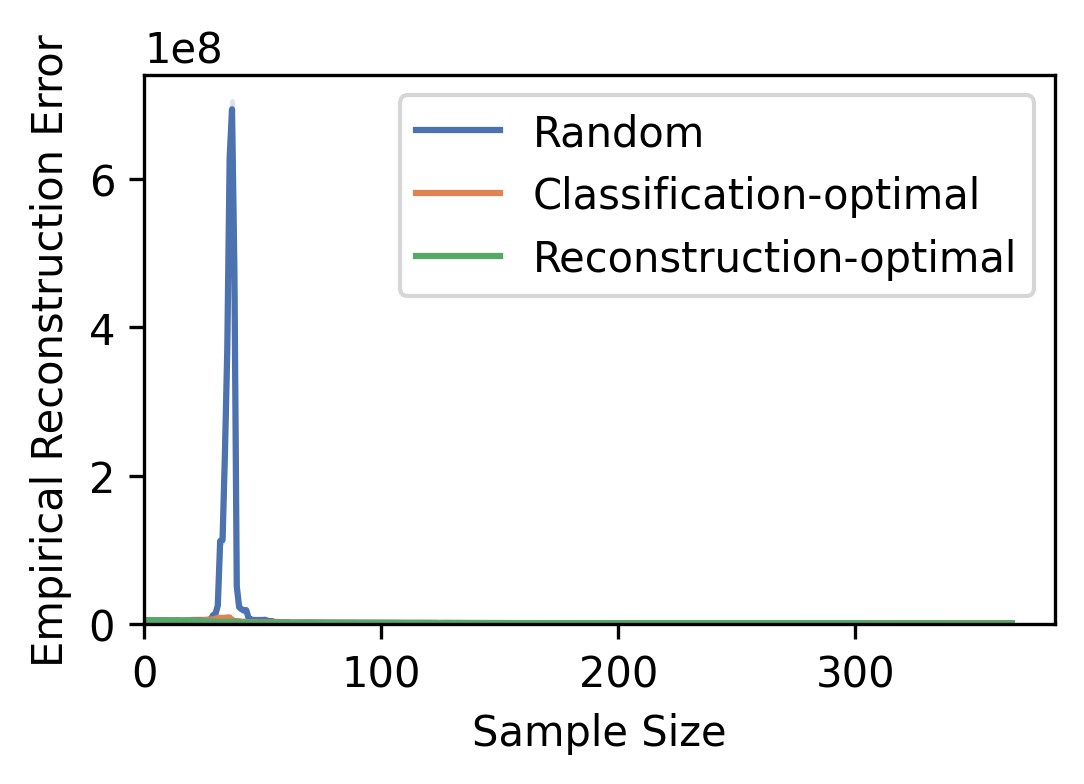}}
    \caption{Rec. Loss (Empirical)}%
    \label{subfig:fmri_LS_rec}
    \end{subfigure}
    %\hfill%
    \begin{subfigure}[T]{0.45\columnwidth}
    \resizebox{\width}{0.62\columnwidth}{
    \includegraphics[width=\columnwidth]{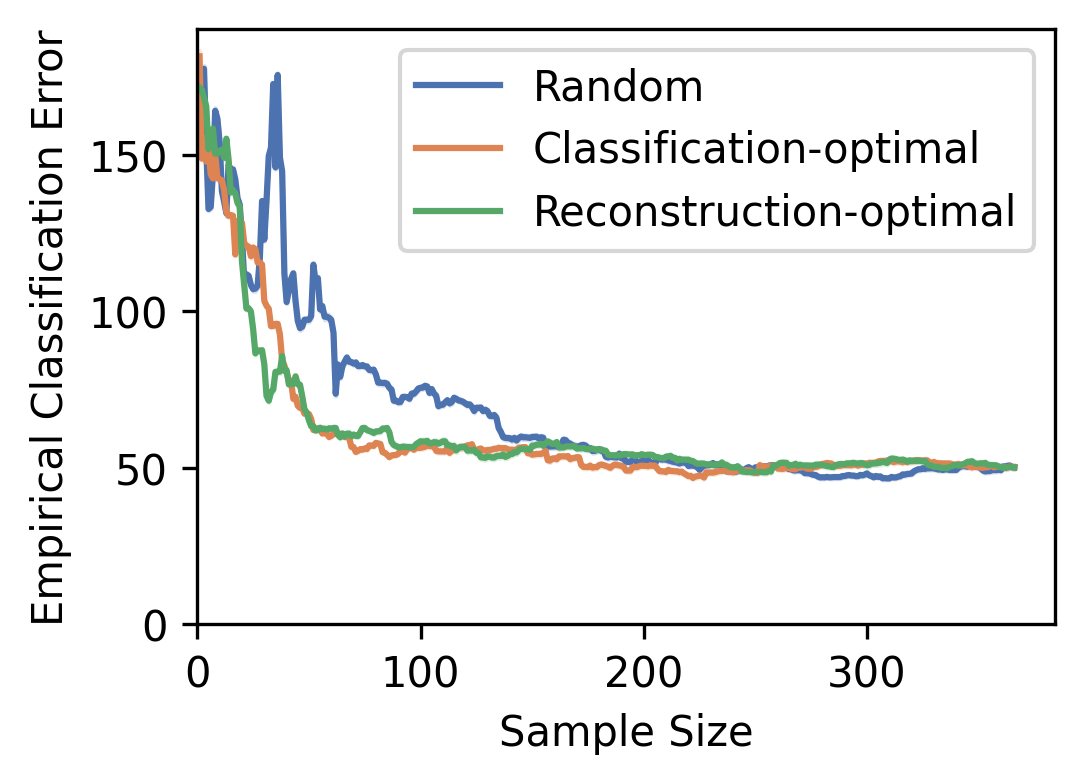}}
    \caption{Class. Loss (Empirical)}
    \label{subfig:fmri_noiseless_sgc_LS_class_anal}
    \end{subfigure}%
    \caption{FMRI dataset, LS reconstruction, sampled assuming $k=36$ and $\text{SNR} = 20dB$}
\label{fig:fmri}
\end{figure}

\subsection{Experimental Validation}
We experimentally validate our derivation of expected classification loss for different graph models, noise levels, reconstruction methods and versions of $f$. Subfigures (b) and (c) of Figs. \ref{subfig:BA_noiseless_sgc_LS}-\ref{subfig:SBM_noiseless_sgc_LS} display the empirical loss calculations via directly reconstructing signals, and the analytic loss calculations using Theorem \ref{thm:general_class_err} respectively. We see that the empirical results largely agree with our theoretical derivations.
%We do not display the $f_{simple}$ cases, which look very similar to the $f_{sgc}$ cases.

We also see that under \bs{noiseless} LS reconstruction (Figs. \ref{subfig:BA_noiseless_sgc_LS} and \ref{subfig:SBM_noiseless_sgc_LS}), our classification loss looks like $x \to \arccos\left(\sqrt{x}\right)$, which follows the theoretical interpretation in \bs{\ref{sec:Gen_Theory}}.

\subsection{Optimal Sampling}
\label{sec:optimal_sampling_results}
In the noiseless setting, under LS reconstruction, our method outperforms both random and A-optimal sampling (Fig. 
\ref{subsubfig:BA_noiseless_sgc_LS_class_anal}) by a small but consistent margin, with no difference in reconstruction error (Fig. \ref{subsubfig:BA_noiseless_sgc_LS_rec}). This margin is even larger in the SBM case (Fig. \ref{subsubfig:SBM_noiseless_sgc_LS_class_anal}). Surprisingly, A-optimal sampling underperforms random sampling; optimizing for reconstruction loss is worse than random sampling under classification loss. \bs{This emphasises how sampling schemes need to be designed in light of the downstream task.}

In the FP reconstruction case (Fig. \ref{subfig:BA_noiseless_sgc_feat_prop}), we see that both forms of optimal sampling are approximately equal, with both beating random sampling significantly. We do see that reconstruction-optimal sampling marginally beats classification sampling under reconstruction loss, and vice-versa for classification loss.

In the noisy case \bs{(Fig. \ref{subfig:BA_noisy_vs_noisy_sgc_LS})}, we see classification-optimal sampling has lower classification loss except for sample sizes near the bandwidth (marked with the vertical dashed line); near the bandwidth, the impact of noise is the highest for LS \cite{sripathmanathan2024impact} and thus reconstruction-optimal sampling, which optimizes for noise sensitivity more aggressively, has both lower reconstruction and classification \bs{loss}. This is an artefact of the greedy construction of our sampling sets. We refer to \cite{sripathmanathan2024impact} for analysis of the spike in reconstruction error in Fig. \ref{subsubfig:BA_noisy_sgc_LS_rec}.

In our real-world dataset \bs{(Fig. \ref{fig:fmri})} we find, for classification loss, both reconstruction and classification-optimal sampling outperforming random sampling, with our novel classification-optimal sampling scheme marginally outperforming at sample sizes over 90.

\section{Conclusion}
In this paper we have studied how classification differs from reconstruction in the presence of noise and missing data. We have given theoretical and empirical characterisation of the relationship between classification and reconstruction error, where classification is done by a linearized GCN. We have shown that reconstruction-optimal sampling can underperform random sampling, which emphasises the point that sampling schemes need to be designed with the specific task in mind. We have then presented a novel optimal sampling scheme specifically for classification with a GCN, and showed its performance across multiple graph models, reconstruction methods and noise levels.

\clearpage
\bibliographystyle{IEEEtran}
\bibliography{combo}

\end{document}